\documentclass[11pt,letterpaper]{article}
\usepackage{epsfig}
\usepackage{fullpage}
\usepackage[hypertex]{hyperref}
\usepackage{fullpage,amsmath,amssymb,amsthm}
\usepackage{xspace}
\usepackage{multirow}
\usepackage{ifpdf}

\newtheorem{theorem}{Theorem}[section]

\newtheorem{lemma}{Lemma}[section]

\def\buildrel#1_#2^#3{\mathrel{\mathop{\kern 0pt#1}\limits_{#2}^{#3}}}

\def\DownPrt2{\mbox{\sc Down\_Part\_2}}

\begin{document}

\title{Label Cover instances with large girth and the hardness of approximating basic $k$-spanner}

\author{
Michael Dinitz \thanks{Department of Computer Science,
Weizmann Institute.
{\tt Michael.Dinitz@weizmann.ac.il}}
\and
Guy Kortsarz
\thanks{Department of Computer Science,
Rutgers, Camden.
{\tt guyk@camden.rutgers.edu}. Partially supported by NSF award number 0829959}
\and
Ran Raz \thanks{Department of Computer Science,
Weizmann Institute.
{\tt ran.raz@weizmann.ac.il}. }
}

\date{\empty}

\maketitle
\begin{abstract}
We study the well-known \emph{Label Cover} problem under the additional requirement that problem instances have large girth.  We show that if the girth is some $k$, the problem is roughly $2^{\log^{1-\epsilon}  n/k}$ hard to approximate for all constant $\epsilon > 0$.
A similar theorem was claimed by Elkin and Peleg [ICALP 2000], but their proof was later found to have a fundamental error.  We use the new proof to show inapproximability for the basic $k$-spanner problem, which is  both the simplest problem in graph spanners and one of the few for which super-logarithmic hardness was not known.  Assuming $NP \not\subseteq BPTIME(2^{polylog(n)})$, we show that for every $k \geq 3$ and every constant $\epsilon > 0$ it is hard to approximate the basic $k$-spanner problem within  a factor better than $2^{(\log^{1-\epsilon} n) / k}$ (for large enough $n$).
A similar hardness for basic $k$-spanner was claimed by Elkin and Peleg [ICALP 2000], but the error in their analysis of Label Cover made this proof fail as well.  Thus for the problem of Label Cover with large girth we give {\em the first} non-trivial  lower bound.  For the basic $k$-spanner problem we improve the previous best lower bound of $\Omega(\log n)/k$ from~\cite{zik}.  Our main technique is subsampling the edges of $2$-query PCPs, which allows us to reduce the degree of a PCP to be essentially equal to the soundness desired.  This turns out to be enough to essentially guarantee large girth.
%
\end{abstract}

\section{Introduction}
In this paper we deal with $2$-query probabilistically checkable proofs (PCPs) and variants of the \emph{Label Cover} problem.  Label Cover was originally defined by Arora and Lund in their early survey on hardness of approximation~\cite{AL96}.  Since then, Label Cover has been widely used as a starting point when proving hardness of approximation, as it corresponds naturally to $2$-query probabilistically checkable proofs and one-round two-prover interactive proof systems. Notable examples are the reduction to the Set Cover problem
\cite{LY94,Feige}, the reduction to the maximum independent set problem
\cite{LY94,Has} and the reduction to the minimum coloring problem
\cite{FK98}.  Certain reductions from Label Cover, though, require special properties of the Label Cover instances.  So then the question becomes: is Label Cover still hard even when restricted to these instances?  For example, the famous \emph{Unique Games Conjecture} of Khot~\cite{khot} can be thought of as a conjecture that Label Cover is still difficult to approximate when the relation on each edge is required to be a bijection.  A different type of requirement is on the structure of the Label Cover graph rather than on the allowed relations; Elkin and Peleg~\cite{ELD} showed that if Label Cover (actually, a slight variant known as Min-Rep) is hard even on graphs with large girth then the \emph{basic $k$-spanner} problem is also hard to approximate.  They then gave a proof that Label Cover is indeed hard to approximate on large-girth graphs, but unfortunately this proof was later found to have a flaw (as Elkin and Peleg point out in~\cite[Section 1.3]{EP07}).  In this paper we give a completely new proof that Label Cover and Min-Rep are hard to approximate on large-girth graphs.  Our argument is based on subsampling edges of a $2$-query PCP/Label Cover instance.  Subsampling of $2$-query PCPs and Label Cover instances has been done before for other reasons (see e.g.~\cite{GS06}), but we show that the sampling probability can be set low enough to destroy most small cycles while still preserving hardness, and this technique was not previously used in this context.  Remaining cycles can then be removed deterministically.

\subsection{Label Cover and Probabilistically Checkable Proofs}
A \emph{probabilistically checkable proof} (PCP) is a string (proof) together with a verifier algorithm that checks the proof (probabilistically).  There are several important parameters of a PCP, including the following:
\begin{enumerate}
\item Completeness ($c$): the minimum probability that the verifier accepts a correct proof.  All of the PCPs in this paper have completeness $1$.
\item Soundness (or error) ($s$): the maximum probability that the verifier accepts a proof of an incorrect claim.
\item Queries: the number of queries that the verifier makes to the proof.  In this paper we will study the case when the verifier only makes $2$ queries.
\item Size: the number of positions in the proof (i.e.~the length).
\item Alphabet ($\Sigma$): the set of symbols that the proof is written in.  We will only be concerned with PCPs for which $|\Sigma|$ is at most polynomial in the size of the PCP, so we will assume this throughout.  
\end{enumerate}

For this paper we will be concerned with $2$-query PCPs, which are the special case when the verifier is only allowed to query two positions of the proof.  We will also assume (without loss of generality) that these two queries are to different parts of the proof, i.e.~there is some set of positions $A$ that can be read by the first query and some other set of positions $B$ that can be read by the second query with $A \cap B = \emptyset$.

For this type of PCP, there are two natural graphs that represent it.  The first (and simpler), which is sometimes called the \emph{supergraph}, is a bipartite graph $(A,B,E)$ in which there is a vertex for every position of the proof and an edge between two positions if there is a possibility that the verifier might query those two positions.  By our assumption, this graph is bipartite.  We also assume that the verifier chooses its query uniformly at random from these edges, which is in fact the case in the PCPs that we will use (in particular in the PCP for Max-3SAT(5) obtained by parallel repetition).  Vertices of this graph will sometimes be called \emph{supervertices}, and edges will be \emph{superedges}

The second graph can be thought of as an expansion of the supergraph in which the test the verifier does is explicitly contained in the graph.  This graph is also bipartite, with vertex set $(A \times \Sigma_A, B \times \Sigma_B)$, where $\Sigma_A$ is the alphabet used in $A$ positions of the proof and similarly for $\Sigma_B$.  There is an edge between vertices $(a, \alpha)$ and $(b, \beta)$ if the verifier might query $a$ and $b$ together (i.e.~$(a,b)$ is an edge in the supergraph) and if, upon such queries, the verifier will accept the proof if it sees values $\alpha$ in position $a$ and $\beta$ in position $b$.  We call this graph the Min-Rep graph to distinguish it from the supergraph. This is related to the work in \cite{zik}.

There is a natural correspondence between these types of PCPs and the optimization problem of Label Cover.  In Label Cover we are given a bipartite graph $G = (A, B, E)$, alphabets $\Sigma_A$ and $\Sigma_B$, and for every edge $e \in E$ a nonempty relation $\pi_e \subseteq \Sigma_A \times \Sigma_B$.  The goal is to find assignments $\gamma_A : A \rightarrow \Sigma_A$ and $\gamma_B : B \rightarrow \Sigma_B$ in order to maximize the number of edges $e = (a,b)$ for which $(\gamma_A(a), \gamma_B(b)) \in \pi_e$ (in which case we say that the edge is satisfied or covered).  It is easy to see that the existence of PCPs for NP-hard problems implies that Label Cover is hard to approximate: in particular, if we use the supergraph and set the relation to be answers on which the verifier accepts, then it is hard to distinguish instances in which at least a $c$ fraction of the edges are satisfiable (a valid proof) from instances in which at most an $s$ fraction of the edges are satisfiable (an invalid proof).  The exact nature of the hardness assumption is based on the size of the PCP: if it has size $m(n)$ (where $n$ is the size of the original problem instance) then the hardness assumption is that NP is not contained in DTIME($poly(m(n))$) (for deterministic PCP constructions and approximation algorithms) or BPTIME($poly(m(n))$) (for randomized PCP constructions or approximation algorithms).

We now describe a slight variant of Label Cover known as Min-Rep (originally defined by Kortsarz~\cite{zik}) that has been useful for proving hardness of approximation for network design problems such as spanners.  It can be thought of as a minimization version of Label Cover with the additional property that the alphabets are represented explicitly as vertices in a graph.  Consider the Min-Rep graph $G' = (A \times \Sigma_A, B \times \Sigma_B, E')$.  For every $i \in A$ let $A_i = \{(i,\alpha) \in A \times \Sigma_A\}$ be the set of vertices in the Min-Rep graph corresponding to vertex $i$ in the Label Cover graph, and similarly for $i \in B$ let $B_i = \{(i, \beta) \in B \times \Sigma_B\}$.  Our goal is to choose a set $S$ of vertices of $G'$ of minimum size so that for every $(i,j) \in E$ there are vertices $(i, \alpha)$ and $(j, \beta)$ in $S$ that are joined by an edge in $E'$.  Such a set is called a \emph{REP-cover}, and the vertices in it are called \emph{representatives}.  Less formally, we can think of Min-Rep as being the problem of assigning to every vertex in the supergraph a \emph{set} of labels/representatives (unlike Label Cover, in which only a single label is assigned) so that for every superedge $(a,b)$ there is at least one label assigned to $a$ and at least one label assigned to $b$ that satisfy the relation $\pi_{(a,b)}$.  Note that in the Min-Rep graph the number of vertices is $|A| \cdot |\Sigma_A| + |B| \cdot |\Sigma_B|$, which means that the size of a Min-Rep instance might be larger than the size of the associated Label Cover instance.  The \emph{supergirth} of a Min-Rep graph is just the girth of the supergraph, i.e.~the girth of the associated Label Cover instance.

Two parameters of PCPs/Label Cover that will be important for us are the \emph{degree} and the \emph{girth}.  The degree of a PCP is the maximum degree of a vertex in the supergraph / associated Label Cover instance.  Similarly, the girth of a PCP is the girth in the supergraph / associated Label Cover instance (recall that the girth of a graph is the length of the smallest cycle).

\subsection{The basic $k$-spanner problem and previous work}

The {\em basic $k$-spanner problem}, also called the {\em minimum size $k$-spanner problem},
is the second main subject in this paper.  In this problem we are given an undirected, unweighted graph $G$ and are asked to find the subgraph $G' = (V,E')$ with the minimum number of edges with the property that
\begin{equation} \label{eq:spanner_def}
\frac{dist_{G'}(u,v)}{dist_G(u,v)}\leq k, \mbox{ for every two vertices }
u,v\in V.
\end{equation}
In the above, the distance between two vertices is just the number
of edges in the shortest path between the two vertices, and $dist_H$ is the distance function on a graph $H$.  Any subgraph $G'$ satisfying \eqref{eq:spanner_def} is called a \emph{$k$-spanner} of $G$, and our goal is to find the $k$-spanner with the fewest edges.

Elkin and Peleg~\cite{ELD} proved that if Min-Rep is hard to approximate when restricted to instances with large supergirth, then the basic $k$-spanner problem is also hard to approximate (the required supergirth of the Min-Rep instances depends on the value $k$).  For completeness, we shall later describe their reduction and proof.

We discuss some previous work on spanner problems.
The concept of graph spanners was first invented first by \cite{first}
in a geometric context. To the best of our knowledge
the spanner problem on general graphs was first invented indirectly
by Peleg and Upfal~\cite{PUP} in their work on small routing
tables. This problem was first explicitly
defined in \cite{pel10}, albeit several papers
claim that this problem was invented in
\cite{PelS} (the two journal versions are from the same year).

Spanners appear in remarkably many applications. Peleg and Upfal~\cite{PUP} showed an application of spanners to maintaining
small routing tables. For further applications toward
this subject see \cite{Cowen, Cowen1,Z1,TZ}.
In \cite{pel10} a relation between
sparse spanners and synchronizers for distributed networks was found.
In \cite{edit1,edit2,el1,Fin}
applications of spanners to parallel, distributed,
and streaming algorithms for shortest paths are described.
For applications of spanners to distance oracles see \cite{sure,MKUZ}.
For applications of spanners in property testing and related subjects see
\cite{BGJRW,arnab,MJSR,RSF}.

For completeness, we note that the basic $k$-spanner problem has been solved for the special case of $k=2$: there is an $O(\log n)$ approximation~\cite{guys1}, and that is the best possible~\cite{zik}.

\subsection{Our results}

All of our results hold for large $n$, so throughout this paper we will assume that $n$ is sufficiently large.  Our first result is on the hardness of approximating Label Cover with large girth:
\begin{theorem}
Assuming $NP \not\subseteq BPTIME(2^{polylog(n)})$, for any constant $\epsilon > 0$ and for $3 \leq k \leq \log^{1-2\epsilon} n$ there is no polynomial-time algorithm that approximates Label Cover with girth greater than $k$ to a factor better than $2^{(\log^{1-\epsilon} n)/k}$.
\end{theorem}

We also show how to adapt this hardness from Label Cover to Min-Rep, which then gives us strong hardness for basic $k$-spanner.

\begin{theorem}
\label{theorem1}
Assuming $NP \not\subseteq BPTIME(2^{polylog(n)})$, for any constant $\epsilon > 0$ and for $3 \leq k \leq \log^{1-2\epsilon} n$ there is no polynomial-time algorithm that approximates Min-Rep with supergirth greater than $k$ to a factor better than $2^{(\log^{1-\epsilon} n)/k}$.
\end{theorem}

\begin{theorem}
\label{theorem2}
Assuming $NP \not\subseteq BPTIME(2^{polylog(n)})$, for any constant $\epsilon > 0$ and for $3 \leq k \leq \log^{1-2\epsilon} n$ there is no polynomial-time algorithm that approximates the basic $k$-spanner problem to a factor better than $2^{(\log^{1-\epsilon} n) / k}$.
\end{theorem}

This is the same hardness for basic $k$-spanner as was claimed by~\cite{ELD}, and so (as they note) is tight at least for the case of $k = \log^{\mu} n$ where $\mu$ is any constant in $(0,1)$.  This is because the best known approximation algorithm for basic $k$-spanner~\cite{ADDJ} gives a ratio of essentially $n^{\Theta(1/k)} = 2^{\Theta(\log^{1-\mu} n)}$, while Theorem~\ref{theorem2} implies hardness of $2^{\log^{1-\mu-\epsilon} n}$ for arbitrarily small constant $0 < \epsilon < (1-\mu)/2$.  

%

\subsection{The error in \cite{ELD}}
To the best of our knowledge the question answered in
Theorem \ref{theorem1} regarding
the hardness of Min-Rep with large supergirth
was first presented in ICALP 2000 by Elkin and Peleg \cite{ELD}.
In \cite{ELD} the authors tried to create Min-Rep instances with large
supergirth that are also hard to approximate as follows.
They started with a $3$-Sat($5$)
graph with large supergirth, which is relatively easy to obtain by standard transformations of 3CNF formulae; the
supergirth is with respect to the graph that has
the clauses and variables as vertices, with an edge between a clause and a variable if the variable is in the clause. They then applied
the
parallel repetition theorem \cite{R} and
claimed to boost the hardness while maintaining the large supergirth.
This reduction is incorrect (as Elkin and Peleg acknowledge in~\cite{EP07}), as non-simple cycles such as paths
in the original graph become simple cycles after parallel repetition is applied.
In fact the supergirth in the construction of \cite{ELD} is $4$, no matter what the initial supergirth before parallel repetition is, and thus \cite{ELD} does not imply any
hardness whatsoever for the large supergirth Min-Rep
problem.  For the interested reader, in the conference version of~\cite{ELD} it is Lemma~13 which is incorrect.

Regarding the hardness of basic $k$-spanner, in \cite{ELD}
a reduction is given from Min-Rep with supergirth larger than $k+1$ to the basic $k$-spanner problem for $k\geq 3$.  While this reduction is correct, since the hardness proof for large supergirth Min-Rep in \cite{ELD} is incorrect the reduction does not actually imply any hardness for basic $k$-spanner.

The actual situation before our paper
is as follows. No hardness whatsoever
was known for the
Min-Rep with large supergirth problem; our hardness result comes
12 years after this question was first posed.
Regarding lower bounds for the basic $k$-spanner problem, this question was first
raised in \cite{zik} in APPROX 1998.
The best lower bound known (before our paper) was $\Omega(\log n)/k$. The improved hardness we give comes
14 years after this question was first posed.

\subsection{Some remarks on our techniques}

The approach taken by Elkin and Peleg was to first create an instance with large supergirth, and then apply parallel repetition to boost the gap.  Unfortunately, parallel repetition destroys the supergirth, bringing it back down to $4$.  Our idea is to apply parallel repetition \emph{first}, boosting the gap, and then randomly sample superedges to sparsify the supergraph.
It turns out, perhaps surprisingly,
 that to a certain extent these random choices {\em do not decrease the gap}.
This may seem non-intuitive at first as usually
the gap is closely related to superdegree
and a smaller superdegree would imply a smaller gap.
This may have been one of the obstacles
in finding a lower bound for Min-Rep with large supergirth.
However, it turns out that it is possible to keep the gap despite the smaller superdegree.

The hardness that we derive this way is actually for the
Label Cover problem with large girth and not for the Min-Rep problem with large supergirth.
The standard reduction from Label Cover to Min-Rep \cite{zik}
entails duplications of many super vertices.
This is needed in order to ensure regularity in the Min-Rep graph, which is used to ensure that removing a $\mu$ fraction of the supervertices
will imply a removal of at most a $\mu$ fraction of the superedges.
In \cite{zik} this duplication is done after the parallel repetition step,
as the supergirth was not an important quantity.   However, such duplications add many cycles of length $4$ in the supergraph.  We handle this difficulty by performing duplication {\em before} we apply parallel repetition.  It is clear that parallel repetition applied to a regular graph maintains the regularity, which is the property that we need to go from the maximization objective of Label Cover to the minimization objective of Min-Rep.

\section{Sampling Lemma for $2$-query PCPs}

We begin with our general $2$-query PCP sampling lemma.  We remark that similar subsampling techniques have been used before (notably by Goldreich and Sudan~\cite{GS06} to give almost-linear size PCPs), but we specialize the technique with an eye towards giving a tradeoff between the soundness and the girth.  Since we will only be concerned with regular PCPs, we will phrase it for the special case when the supergraph has $|A| = |B| = n/2$ and is regular with degree $d$.  We will assume without loss of generality that $|\Sigma_A| \geq |\Sigma_B|$. Given such a PCP verifier (i.e.~Label Cover instance) $G = (A,B, E)$, let $G_{\alpha}$ be the verifier/instance that we get by sub-sampling the edges with probability $p_{\alpha} = \frac{\alpha \log |\Sigma_A|}{d}$, i.e.~every edge $e \in E$ is included in $G_{\alpha}$ independently with probability $p_{\alpha}$.

\begin{lemma} \label{lem:sample_soundness1}
Let $G = (A,B,E)$ be a $2$-query PCP verifier/Label Cover instance with completeness $1$ and soundness $s$ in which $|A| = |B| = n/2$, every vertex has degree $d$, and $|\Sigma_A| \geq |\Sigma_B|$.  Let $1 \leq \alpha \leq 1/s$.  Then with high probability $G_{\alpha}$ is a PCP verifier with completeness $1$ and soundness at most $4e / \alpha$.
\end{lemma}
\begin{proof}
It is obvious that $G_{\alpha}$ has completeness $1$ with probability $1$, since any valid labeling/proof of $G$ is also valid for $G_{\alpha}$.  To bound the soundness, first fix a proof / labeling.  We know that in the original verifier, at most an $s$ fraction of the edges are satisfied.  After sampling, the expected number of satisfied edges is at most
\begin{equation*}
p_{\alpha} s |E| \leq \frac{|E| \log |\Sigma_A|}{d} = \frac{n}{2} \log |\Sigma_A|.
\end{equation*}
Since the sampling decisions are independent we can apply a Chernoff bound (see e.g.~\cite[Theorem 1.1]{DP09}), giving us that the probability that more than $en \log |\Sigma_A|$ edges are satisfied is at most $2^{-en \log |\Sigma_A|} = |\Sigma_A|^{-en}$.  But the total number of possible proofs is at most $|\Sigma_A|^{n/2} |\Sigma_B|^{n/2} \leq |\Sigma_A|^n$.  So by a union bound, the probability that any labeling satisfies more than $en \log |\Sigma_A|$ edges is at most $|\Sigma_A|^{-(e-1)n} \leq 2^{-n}$, which is negligible.  But the expected total number of edges after sampling is $p_{\alpha}|E| = \frac{n}{2} \alpha \log |\Sigma_A|$, and so another Chernoff bound implies that with high probability the number of edges after sampling is at least $(n/4) \alpha \log |\Sigma_A|$.  Thus with high probability no proof is accepted with probability more than $(en \log |\Sigma_A|) / ((n/4) \alpha \log |\Sigma_A|) = 4e / \alpha$.
\end{proof}

Lemma~\ref{lem:sample_soundness1} shows that we can  sample edges so that the average degree is about $\alpha \log |\Sigma|$ without hurting the soundness too much (in particular, the soundness becomes basically $1/\alpha$).  Note that if we set $\alpha = 1/s$ this allows us to reduce the average degree to basically $(1/s) \log |\Sigma_A|$ (a possibly significantly reduction) without affecting the soundness by more than a constant.  We would like to claim that this lets us increase the girth, but at this point we will merely prove that any edge is \emph{unlikely} to be in short cycles.  Later we will deterministically remove edges that take part in short cycles, but since that might destroy approximate-regularity (which is necessary for our reduction to Min-Rep) we put it off until later.

\begin{lemma} \label{lem:edge_cycle}
Fix an edge $(u,v) \in G$.  Conditioned on $(u,v) \in G_{\alpha}$, the probability that $(u,v)$ is in a cycle in $G_{\alpha}$ of length at most $k$ is at most $\frac{2(\alpha \log |\Sigma_A|)^{k-1}}{d}$.
\end{lemma}
\begin{proof}
Let $4 \leq k' \leq k$ (note that no edge is in a cycle of length less than $4$ in any bipartite graph, including $G$).  Fix a cycle of length $k'$ in $G$ that contains $(u,v)$.  After conditioning on $(u,v)$ surviving the sampling, the probability that all of the other edges in the cycle are also in $G_{\alpha}$ is $p_{\alpha}^{k'-1} = \left(\frac{\alpha \log |\Sigma_A|}{d}\right)^{k'-1}$.  On the other hand, we know from the degree bound in $G$ that the number of $k'$-cycles containing $(u,v)$ is at most $d^{k'-2}$.  So a union bound implies that the probability that $(u,v)$ is in a $k'$-cycle in $G_{\alpha}$ is at most $\frac{(\alpha \log |\Sigma_A|)^{k'-1}}{d}$.  Now we take a union bound over all $4 \leq k' \leq k$ to get that the total probability that $(u,v)$ is in a cycle of length at most $k$ is at most $\sum_{k'=4}^k \frac{(\alpha \log |\Sigma_A|)^{k'-1}}{d} \leq \frac{2(\alpha \log |\Sigma_A|)^{k-1}}{d}$ as claimed (assuming without loss of generality that $\alpha \log |\Sigma_A| \geq 2$).
\end{proof}

It is easy to see that subsampling preserves approximate regularity, but we will now prove so formally.

\begin{lemma} \label{lem:sampling_regularity}
If $\alpha \geq 16 \log n$ then with probability at least $1-2/n$ the degree of every vertex in $G_{\alpha}$ is between $\frac12 \alpha \log |\Sigma_A|$ and $2 \alpha \log |\Sigma_A|$.
\end{lemma}
\begin{proof}
Since $G$ is regular with degree $d$ and $p_{\alpha} = \frac{\alpha \log |\Sigma_A|}{d}$, the expected degree of a vertex $v$ in $G_{\alpha}$ is clearly $\alpha \log |\Sigma_A|$.  So by a Chernoff bound (see e.g.~\cite[Theorem 1.1]{DP09}), the probability that the degree is less than $1/2$ of this or more than twice this is at most $2 \cdot e^{-\alpha \log |\Sigma_A| / 8} \leq 2 / |\Sigma_A|^{-\alpha / 8}$.  Since $\alpha \geq 16 \log n$ and $|\Sigma_A \geq 2$, this probability is at most $2/n^2$, so taking a union bound over vertices implies that the probability that all of them have degree in the desired range is at least $1 - 2/n$.
\end{proof}

\section{Label Cover and Min-Rep with large (super)girth}

In this section we show that Label Cover and Min-Rep are both hard to approximate, even when restricted to instances with large (super)girth.  We start with a PCP verifier with supergraph $G$ and Min-Rep graph $H$, and then use the previously described random sampling technique to get a new supergraph $G_{\alpha}$ and Min-Rep graph $H_{\alpha}$.  We now remove from $G_{\alpha}$ any edge that is in a cycle of length at most $k$, giving us a new supergraph $G'_{\alpha}$ and Min-Rep graph $H'_{\alpha}$ (where an edge $((a, \delta), (b, \beta))$ from $H$ is in $H'_{\alpha}$ if $(a,b)$ remains as an edge in $G'_{\alpha}$).  These instances will form our reduction.

We say that an edge $(i,j) \in G_{\alpha}$ is \emph{bad} if it is not in $G'_{\alpha}$, i.e.~it is part of a cycle of length of at most $k$ in $G_{\alpha}$.

\begin{lemma} \label{lem:large_girth}
Let $16 \log n \leq \alpha \leq \min\{1/s, d^{1/k} / \log |\Sigma_A|\}$.  Then with probability larger than $2/3$ the number of bad edges is at most $O\left(\frac{n(\alpha \log |\Sigma_A|)^k}{d}\right) \leq O(n)$
\end{lemma}
\begin{proof}
Lemma~\ref{lem:edge_cycle} and Markov's inequality imply that with probability at least $3/4$, at most a $\frac{8(\alpha \log |\Sigma_A|)^{k-1}}{d}$ fraction of the edges are bad.  With high probability the total number of edges in $G_{\alpha}$ is $\Theta(n \alpha \log |\Sigma_A|)$, so this means that the number of bad edges is at most $O\left(\frac{n (\alpha \log |\Sigma_A|)^k}{d}\right)$.  By our choice of $\alpha$, this is at most $O(n)$.  
\end{proof}


\begin{theorem} \label{thm:LC1}
If there is no (randomized) polynomial time algorithm that can distinguish between instances of Label Cover in which $|A| = |B| = n/2$ and all vertices have degree $d$ where all edges are satisfiable and instances where at most an $s$ fraction of the edges are satisfiable, then there is some constant $c$ so that for $18 \log n \leq \alpha \leq \min\{1/s, d^{1/k} / \log |\Sigma_A|\}$ there is no (randomized) polynomial time algorithm that can distinguish between instances of Label Cover with girth larger than $k$ in which all edges are satisfiable and instances in which at most a $c/\alpha$ fraction of the edges are satisfiable.
\end{theorem}
\begin{proof}
If there is a labeling that satisfies all edges of $G$, then clearly the same labeling satisfies all edges of $G'_{\alpha}$.  On the other hand, suppose that only an $s$ fraction of the edges of $G$ can be satisfied.  By Lemma~\ref{lem:large_girth}, the number of bad edges is at most $O(n)$, so the total number of edges in $G'_{\alpha}$ is still $\Theta(n \alpha \log |\Sigma_A|)$.

Fix any labeling of $G'_{\alpha}$, and suppose that it satisfies a $\beta$ fraction of the edges of $G'_{\alpha}$.  Then even if it would have satisfied all of the bad edges, the number of edges of $G_{\alpha}$ that it satisfies is at most $\beta  \cdot \Theta(n \alpha \log |\Sigma_A|) + O(n)$.  By Lemma~\ref{lem:sample_soundness1} this must be at most $(4e/\alpha) \cdot \Theta(n \alpha \log |\Sigma_A|)$, and thus for some constant $c$ we have that $\beta \leq c/\alpha$.

Thus if we could distinguish between the case when every edge of $G'_{\alpha}$ can be satisfied and the case when at most a $c/\alpha$ fraction can be satisfied, we could distinguish between the case when every edge of $G$ can be satisfied and the case when at most an $s$ fraction can be satisfied.
\end{proof}

We now reduce to Min-Rep.  We first show how the size of the minimum REP-cover of $H_{\alpha}$ depends on $G$.  We will then show that, similar to Label Cover, we can deterministically remove small cycles to get the instance $H'_{\alpha}$ with large supergirth that preserves this gap.


\begin{lemma} \label{lem:MR_reduction}
Let $18 \log n \leq \alpha \leq 1/s$.  If there is a valid labeling of $G$ then the Min-Rep instance $H_{\alpha}$ has a REP-cover of size $n$ (where $n$ is the number of vertices in the supergraph).  Otherwise, with high probability the smallest REP-cover has size at least $\Omega(n\sqrt{\alpha})$.
\end{lemma}
\begin{proof}
If there is a valid labeling of $G$ then by Lemma~\ref{lem:sample_soundness1} there is a valid labeling of $G_{\alpha}$ (since the completeness remains $1$), and thus there is a REP-cover of $H_{\alpha}$ of size $n$ as claimed.  On the other hand, suppose that at most an $s$ fraction of the edges of $G$ can be satisfied.  Then since the soundness of $G_{\alpha}$ is at most $4e/\alpha$ by Lemma~\ref{lem:sample_soundness1}, any labeling of $G_{\alpha}$ satisfies at most a $4e / \alpha$ fraction of the edges.  Suppose that there is a REP-cover of $H_{\alpha}$ of size less than $n \sqrt{\alpha} / (36\sqrt{3e})$.  We will show that this implies that there is a labeling of $G_{\alpha}$ that satisfies more than a $4e / \alpha$ fraction of the edges, giving a contradiction and proving the lemma.

Suppose that the smallest REP-cover for $H_{\alpha}$ has size $\beta n$.   This means that the \emph{average} number of representatives/labels assigned to each vertex in $G_{\alpha}$ by this cover is $\beta$.  To analyze the labeling that covers the most edges, we analyze the random labeling obtained by choosing for each vertex a label uniformly at random from the set of labels it is assigned by the REP-cover.  Let $A' \subseteq A$ be the set of vertices in $A$ that receive at most $18 \beta$ labels in this REP-cover, and define $B' \subseteq B$ analogously.  Note that $|A'| \geq (8/9) |A|$ and similarly $|B'| \geq (8/9) |B|$, since otherwise the total number of representatives in the REP-cover is larger than $(1/9) \cdot (n/2) \cdot (18 \beta) = \beta n$, contradicting our assumption on the size of the REP-cover.  With high probability the fraction of edges of $G_{\alpha}$ that touch a vertex of $A \setminus A'$ is at most $\frac{2d (1/9)}{(2d (1/9) + (8/9) (d/2)} = 1/3$, and similarly for the fraction of edges that touch $B \setminus B'$ (where we used the approximate regularity from Lemma~\ref{lem:sampling_regularity}).  So if we consider the subgraph of $G_{\alpha}$ induced by $A' \cup B'$ we still have at least $1/3$ of the edges of $G_{\alpha}$.

Now let $(a,b) \in A' \times B'$ be such an edge.  Since we started with a REP-cover, there is at least one representative assigned to $a$ and one representatives assigned to $b$ that satisfy the relation $\pi_{(a,b)}$.  Since both endpoints have at most $18\beta$ representatives in the REP-cover, the probability that these two representatives are the assigned labels is at least $1/(18\beta)^2$.  Thus by linearity of expectations we expect that at least $1/(3 (18\beta)^2) = 1/(972 \beta^2)$ fraction of the edges are covered by our random labeling, so there exists some labeling that covers at least that many.  If $\beta \leq \frac{\sqrt{\alpha}}{36 \sqrt{3e}}$ then this is at least $4e / \alpha$, giving a contradiction.  Thus the smallest REP-cover has size at least  $(n \sqrt{\alpha}) / (36\sqrt{3e})$, proving the lemma.
\end{proof}

We will now get rid of small cycles by using the instance $H'_{\alpha}$, proving hardness of Min-Rep with large supergirth.

\begin{theorem} \label{thm:main1}
If there is no (randomized) polynomial time algorithm that can distinguish between instances of Label Cover in which $|A| = |B| = n/2$ and all vertices have degree $d$ where all edges are satisfiable and instances where at most an $s$ fraction of the edges are satisfiable, then there is some constant $c$ so that for $18 \log n \leq \alpha \leq \min\{1/s, d^{1/k} / \log |\Sigma_A|\}$ there is no (randomized) polynomial time algorithm that can distinguish between instances of Min-Rep with supergirth larger than $k$ where the smallest REP-cover has size at most $n$ and instances where the smallest REP-cover has size at least $n \sqrt{\alpha} / c$ (here $n$ is the size of the supergraph).
\end{theorem}
\begin{proof}
If there is a labeling that satisfies all edges of $G$, then clearly choosing that labeling gives a valid REP-cover of $H'_{\alpha}$ of size at most $n$.  For the other case, suppose that any labeling of $G$ satisfies at most an $s$ fraction of the edges.  Then by Lemma~\ref{lem:MR_reduction}, with high probability the smallest REP-cover of $H_{\alpha}$ has size at least $\Omega(n\sqrt{\alpha})$.
By Lemma~\ref{lem:large_girth}, the number of bad edges is at most $O(n)$.  Removing any particular edge (in particular a bad edge) can only decrease the size of the optimal REP-cover by at most $2$, so if we remove all bad edges (getting $H'_{\alpha}$) we are left with an instance with supergirth larger than $k$ in which the smallest REP-cover has size at least $\Omega(n \sqrt{\alpha}) - O(n) = \Omega(n \sqrt{\alpha})$.  By construction the supergirth is greater than $k$, so this proves the theorem.
\end{proof}

Now we define and analyze the PCP / Label Cover instances to which we will apply Theorems~\ref{thm:LC1} and~\ref{thm:main1}.  Recall that Max-3SAT(5) is the problem of finding an assignment to variables of a 3-CNF formula that maximizes the number of satisfied clauses, with the additional property that every variable appears in exactly 5 clauses of the formula.   We begin with the standard Label Cover instance for Max-3SAT(5) (see for example~\cite{AL96}).  The graph $(A,B,E)$ has $|B| = n'$ and $|A| = 5n'/3$ (where $n'$ is the number of variables in the instance), and every vertex in $A$ has degree $3$ and every vertex in $B$ has degree $5$.  Vertices in $A$ correspond to clauses and vertices in $B$ correspond to variables.  The alphabet sizes are $|\Sigma_A| = 7$ and $|\Sigma_B| = 2$.  The PCP Theorem~\cite{SSAA,AS98} implies that the gap for these instances is constant, i.e.~it is hard to distinguish the case when all edges are satisfiable from the case in which $1/(1+\epsilon)$ fraction of the edges are satisfiable, for some constant $\epsilon$.

   Now we take three copies of $A$, call them $A_1, A_2, A_3$, and let $A'$ be their union (so $|A'| = 5n'$).  Similarly we take five copies of $B$ to get $B_i$ for $i \in [5]$, and take their union to be $B'$.  Now between each $A_i$ and each $B_j$ we put a copy of the original edge set $E$ (which we will call $E_{ij}$), giving us edge set $E'$.  Note that $|B'| = |A'| = 5n'$ and every vertex has degree $15$.  Obviously if the original instance has all edges satisfiable then that is still true of this instance.  On the other hand, suppose in the original instance at most $1/(1+\epsilon)$  of the edges are satisfiable.
Then fix any labeling of $A'$ and $B'$.  This induces some labeling of $A_i$ and $B_j$, which we know can only satisfy $1/(1+\epsilon)$ of the edges in $E_{ij}$.  Since this is true for all $i,j$, the total fraction of satisfied edges is at most $1/(1+\epsilon)$.

We now apply parallel repetition $\ell$ times.  Now each side has size $(5n')^{\ell}$, the degree is $d = 15^{\ell}$, and the alphabet sizes are $|\Sigma_A| = 7^{\ell}$ and $|\Sigma_B| = 2^{\ell}$.  By the parallel repetition theorem~\cite{R}, unless $NP \subseteq BPTIME(n^{O(\ell)})$ it is hard to distinguish between the case when all edges are satisfiable and when at most a $2^{-\ell / c}$ fraction are satisfiable for some constant $c$.  We can apply Theorem~\ref{thm:LC1} to this construction to get the following hardness result.

\begin{theorem} \label{thm:LC_main}
Assuming $NP \not\subseteq BPTIME(2^{polylog(n)})$, for any constant $\epsilon > 0$ and $3 \leq k \leq \log^{1-2\epsilon} n$ there is no polynomial time algorithm that can approximate Label Cover with girth greater than $k$ to a factor better than $2^{(\log^{1-\epsilon} n) / k}$.
\end{theorem}
\begin{proof}
Set $\ell = \log^{1/\epsilon} n'$, so the size of the Label Cover instance is $n = (5n')^{\log^{1/\epsilon} n'}$ and $\ell^{\epsilon} = \log n'$.   Note that $\log n = \Theta(\ell \log n') = \Theta(\ell^{1+\epsilon})$, so $\ell = \Theta((\log n)^{1/(1+\epsilon)})$.  Let $\alpha = \min\{2^{\ell / c}, 15^{\ell / k} / \ell \log 7\}$. Assuming that $k$ is at most $\log^{(1/(1+\epsilon))- \gamma} n$ for some constant $\gamma > 0$ implies that $\alpha \geq 16 \log n$, so applying Theorem~\ref{thm:LC1} to this construction implies that, assuming $NP \not\subseteq BPTIME(n^{O(\ell)})$, there is no polynomial time algorithm that can distinguish between instances of Label Cover with girth greater than $k$ in which all edges are satisfiable and instances in which at most a $c/\alpha$ fraction are satisfiable (for some constant $c$).  Using a smaller $\epsilon$ to change $1/(1+\epsilon)$ to $1-\epsilon$, as well as to get rid of lower order terms, gives the theorem.
\end{proof}

On the other hand, if we apply Theorem~\ref{thm:main1} to this construction then we get the following theorem:

\begin{theorem}
\label{maintheorem}
Assuming $NP \not\subseteq BPTIME(2^{polylog(n)})$, for any constant $\epsilon > 0$ and $3 \leq k \leq \log^{1-2\epsilon} n$ there is no polynomial time algorithm that can distinguish between instances of Min-Rep with supergirth greater than $k$ that have a REP-cover of size at most $\tilde n$ and instances in which the smallest REP-cover has size at least $2^{(\log^{1-\epsilon} n) / k} \cdot \tilde n$, where $n$ is the size of the Min-Rep graph and $\tilde n$ is the size of the supergraph.  Thus there is no polynomial time algorithm that can approximate Min-Rep with supergirth greater than $k$ to a factor better than $2^{(\log^{1-\epsilon} n) / k}$.
\end{theorem}
\begin{proof}
As before, we set $\ell = \log^{1/\epsilon} n'$ (so $\ell^{\epsilon} = \log n'$).  Then $n = (5n')^{\ell} \cdot 7^{\ell} + (5n')^{\ell} \cdot 2^{\ell} \leq 2 (35n')^{\ell}$ is the number of vertices in the resulting Min-Rep instance.  Note that, as in Label Cover, $\log n = \Theta(\ell \log n') = \Theta(\ell^{1+\epsilon})$, so $\ell = \Theta((\log n)^{1/(1+\epsilon)})$.  Applying Theorem~\ref{thm:main1} to this construction implies that unless $NP \subseteq BPTIME(2^{polylog(n)})$, there is no polynomial time algorithm that can distinguish between instances of Min-Rep with supergirth larger than $k$ where the smallest REP-cover has size at most $2(5n')^{\ell}$ and instances where the smallest REP-cover has size at least $\Omega\left((5n')^{\ell} \sqrt{\min\{2^{\ell / c}, 15^{\ell / k} / \ell \log 7\}}\right)$ (assuming $k \leq \log^{(1/(1+\epsilon)) - \gamma} n$ for some constant $\gamma > 0$).  Plugging in our choice of $\ell$, and using smaller values of $\epsilon$ to get rid of lower order terms and replace $1/(1+\epsilon)$ by $1-\epsilon$, gives the theorem.
\end{proof}

\section{Hardness of basic $k$-spanner}

We now show how to reduce Min-Rep with large supergirth to the basic $k$-spanner problem.  This reduction is from Elkin and Peleg~\cite{ELD}, which is in turn very similar to reductions from Min-Rep to other spanner problems developed by Elkin and Peleg~\cite{EP07}.  We include it here just for completeness, and try to use their notation when possible.  Suppose that we are given a Min-Rep instance with supergraph $\tilde G= (A,B,\tilde E)$ with supergirth at least $k+2$, as well as the Min-Rep graph $G=(A \times \Sigma_A, B \times \Sigma_B, E)$.  As before, for $i \in A$ let $A_i = \{(i, \alpha) : \alpha \in \Sigma_A\}$ be the set of vertices in the Min-Rep graph corresponding to the vertex $i$ in the supergraph, and similarly for $i \in B$ let $B_i = \{(i, \beta) : \beta \in \Sigma_B\}$.  Let $n = |A| \cdot |\Sigma_A| + |B| \cdot |\Sigma_B|$ denote the size of the Min-Rep graph, and let $\tilde n = |A| + |B|$ denote the size of the supergraph.  Since this instance comes from our previous reduction, we can also assume that $|A| = |B| = \tilde n / 2$.  Let $k_A = \lfloor \frac{k-1}{2}\rfloor$, let $k_B = \lceil \frac{k-1}{2} \rceil$, and let $x = n^2 / \tilde n$.  To define the $k$-spanner graph $G'$, we first define two new vertex sets:
\begin{equation*}
S = \{s^p_{ij} : i \in A, j \in [k_A], p \in [x]\} \text{ and } T = \{t^p_{ij} : i \in B, j \in [k_B], p \in [x]\}.
\end{equation*}

The vertex set of our graph $G'$ will be $V' = A \cup B \cup S \cup T$.  Now we define a collection of edge sets:
\begin{align*}
E_M &= \{(s_{ij}^p, s^p_{i(j+1)}) : p \in [x], i \in A, j \in [k_A - 1]\}, \\
& \ \ \ \ \cup \{(t^p_{ij}, t^p_{i(j+1)}) : p \in [x], i \in B, j \in [k_B - 1]\}, \\
E_{sA} &= \{(s^p_{i1}, u) : i \in A, u \in A_i, p \in [x]\}, \\
E_{tB} &= \{(w, t^p_{j1}) : j \in B, w \in B_j, p \in [x]\}, \\
E_{\tilde G}^p &= \{(s^p_{i k_A}, t^p_{j k_B}) : i \in A, j \in B, (i,j) \in \tilde E\}, \\
E_{\tilde G} &= \cup_{p=1}^x E_{\tilde G}^p.
\end{align*}
We let the edge set $E'$ of $G'$ be $E \cup E_M \cup E_{sA} \cup E_{tB} \cup E_{\tilde G}$.   Note are that when $k=3$, $k_A = k_B = 1$, so $E_M$ is empty.  Also note that for each $p \in [x]$, the edges $E^p_{\tilde G}$ form a graph isomorphic to the supergraph $\tilde G$.

For an edge $(s^p_{i k_A}, t^p_{j k_B}) \in E^p_{\tilde G}$, we say that a spanning path $P$ (i.e.~a path from $s^p_{i k_A}$ to $t^p_{j k_B}$ of length at most $k$) is a \emph{canonical path} if it has the form $s^p_{i k_A}, s^p_{i (k_A-1)}, \dots, s^p_{i 1}, u_i, w_j, t^p_{j 1}, t^p_{j 2}, \dots, t^p_{j k_B}$.  In other words, the path first follows the path of $E_M$ edges to $s^p_{i 1}$, then uses an edge from $E_{sA}$ to get to one of the original Min-Rep nodes $u_i$ that corresponds to supernode $i$, then takes an original Min-Rep edge to $w_j$, then an $E_{tB}$ edge out to $t^p_{j 1}$, and then follows $E_M$ edges to $t^p_{j k_B}$.  Note that such a path must exist, or else there are no edges from $A_i$ to $B_j$, in which case $(i,j)$ would not be an edge in the supergraph, which would mean that $(s^p_{i k_A}, t^p_{j k_B})$ would not be an edge in $E^p_{\tilde G}$.  Furthermore, any canonical path has length exactly $(k_A-1) +1+1+1+(k_B-1) = k$, so is indeed a valid spanning path.

\begin{lemma} \label{lem:span1}
In any $k$-spanner $H$ of $G'$, every edge in $E_{\tilde G}$ is either included in $H$ or is spanned by a canonical path.
\end{lemma}
\begin{proof}
Suppose this is false, and let $e = (s^p_{i k_A}, t^p_{j k_B})$ be an edge in $E_{\tilde G}$ which is not in $H$ but is also not spanned by a canonical path.  Let $P = s^p_{i k_A} = x_1, x_2,  \dots, x_{q-1}, x_q = t^p_{j k_B}$ be the shortest simple path in $H$ that does span $e$ (such a path with $q \leq k+1$ must exist since $H$ is a $k$-spanner of $G'$).  Let $U = \{s^p_{i' k_A} : i' \in A\} \cup \{t^p_{j' k_B} : j' \in B\}$ be the set of vertices that are incident on edges of $E^p_{\tilde G}$.  Let $x_{\alpha}$ be the first vertex in $P$ that is not in $U$.  Such a vertex must exist, since if it does not then $P$ is a path of length at most $k$ between $s^p_{i k_A}$ and $t^p_{j k_B}$ that is contained in $E^p_{\tilde G}$.  This is a contradiction, since $E^p_{\tilde G}$ is isomorphic to the supergraph $\tilde G$ which by assumption has girth at least $k+2$, while adding $e$ to $P$ would give a cycle of length at most $k+1$.

So $x_{\alpha}$ is the first vertex in $P$ that is not in $U$.  This means that it must be either $s^p_{i'(k_A-1)}$ or $t^p_{j'(k_B-1)}$ for some $i' \in A$ or $j' \in B$.  If it is $t^p_{j'(k_B-1)}$, then $P$ must keep following $E_M$ edges to get that $x_{\alpha + k_B-2}$ is $t^p_{j'1}$ and thus $x_{\alpha + k_B - 1}$ is $w$ for some $w \in B_{j'}$.  Since $\alpha \geq 2$ and $k_A + k_B = k-1$, there can be only $k_A+1$ more vertices on the path.  But it is obvious that, if $j' \neq j$ (which it must be since otherwise the path would have already hit $t^p_{j k_B}$), there is no way to complete the path.

If $x_{\alpha} = s^p_{i'(k_A-1)}$, then since the intermediate vertices have degree $2$ and $P$ is simple it must be the case that $x_{\alpha+k_A-2}$ is $s^p_{i'1}$, and so $x_{\alpha + k_A - 1}$ is $u$ for some $u \in A_{i'}$.  From $u$ the next vertex could either be $s^{p'}_{i'1}$ for some other $p'$, or could be $w \in B_{j'}$ for some $j' \in B$.  If $x_{\alpha + k_A}$ is $s^{p'}_{i'1}$ for some $p'$, then the next hop cannot be to a node in $A_{i'}$, or else we could have gotten to this node instead of to $u$ earlier, contradicting our assumption that $P$ is the shortest path.  Thus if $x_{\alpha + k_A}$ is $s^{p'}_{i'1}$ for some $p'$, it must be that $P$ follows $E_M$ edges backwards to get that $x_{\alpha+ 2k_A - 1} = s^{p'}_{i' k_A}$.  Note that $\alpha \geq 2$ by definition, and $2k_A \geq k-2$, so either $x_k$ or $x_{k+1}$ is $s^{p'}_{i' k_A}$.  Either one is an obvious contradiction, since $x_{k+1}$ is supposed to be $t^p_{j k_B}$, which is not adjacent to $s^{p'}_{i' k_A}$.

This means that from $u$ the next vertex in $P$ must be $w \in B_{j'}$ for some $j' \in B$, or equivalently that $x_{\alpha + k_A} = w$.  Since $\alpha \geq 2$ and $k_A + k_B = k-1$,  there can be at most $k_B$ more vertices on $P$.  In order to get to $t^p_{j k_B}$ via $E_M$ edges, it obviously must be the case that $j' = j$ and $P$ is actually a canonical path.  Otherwise, if the next hop from $w$ is another edge from $E$ then it is back on the $A$ side of the Min-Rep graph, which is clearly too far away from $t^p_{j k_B}$ to finish the path.  Thus $P$ must actually be a canonical path.
\end{proof}

We will now define some edge sets that will be useful in the next lemma.  For each $i \in A$, let $u_i \in A_i$ be some arbitrarily chosen vertex in $A_i$, and let $\hat E_i = (\cup_{u \in A_i} (s^1_{i1}, u)) \cup (\cup_{p \in [x]} (s^p_{i1}, u_i))$.  Similarly, for each $j \in B$ we arbitrarily choose some node $w_j \in B_j$, and let $\hat E_j = (\cup_{w \in B_j} (w, t^1_{j1})) \cup (\cup_{p \in [x]} (w_j, t^p_{j1}))$.  Let $\hat E = (\cup_{i \in A} \hat E_i) \cup (\cup_{j \in B} \hat E_j)$.  Clearly $|\hat E| = n + x \tilde n$, since $|\hat E_i| = |A_i| + x$ for each $i \in A$ and $|\hat E_j| = |B_j| + x$ for each $j \in B$.

We say that a spanner $H$ of $G'$ is a \emph{proper} $k$-spanner if it does not include any edge of $E_{\tilde G}$, which by Lemma~\ref{lem:span1} implies that every edge of $E_{\tilde G}$ is spanned by a canonical path.

\begin{lemma} \label{lem:proper}
Any $k$-spanner $H$ for $G'$ can be converted in polynomial time into a proper $k$-spanner $H'$ of size at most $6 |H|$.
\end{lemma}
\begin{proof}
We first let $H_1$ be the edge set $(H \setminus E_{\tilde G}) \cup E \cup E_M \cup \hat E$.  It is obvious that all edges of $G'$ are $k$-spanned by $H_1$ except for $H \cap E_{\tilde G}$: edges in $E, E_M$, and $\hat E$ are self-spanned, edges in $E_{sA}$ and $E_{tB}$ are $3$-spanned by $\hat E$, and edges in $E_{\tilde G} \setminus H$ must have been spanned in $H$ by a canonical path (by Lemma~\ref{lem:span1}), which is still contained in $H_1$.  We now claim that $H_1$ is small.  Note that $|V'| = n + x(\tilde n / 2) k_A + x(\tilde n /2) k_B = n + x(\tilde n / 2) (k-1)$.  So
\begin{align*}
|H_1| &\leq |H| + |E| + |E_M| + |\hat E| \\
&\leq |H| + n^2 + x (\tilde n/2) (k_A-1) + x (\tilde n /2) (k_B-1) + n + \tilde n x \\
& \leq |H| + x \tilde n + x(\tilde n / 2) (k-3) + n + \tilde n x \\
& \leq |H| + (|V'|-1) + (|V'|-1) + (|V'|-1) \\
& \leq 4|H|,
\end{align*}
where we used the fact that $k \geq 3$ and the fact that $|H| \geq |V'| - 1$ (since it is connected).

Now we need to span the edges in $H \cap E_{\tilde G}$.  For each such edge $(s^p_{i k_A}, t^p_{j k_B})$ there is an associated superedge $(i,j)$.  We get $H'$ from $H_1$ by adding, for each such edge, the edges $(s^p_{i 1}, u)$ and $(w, t^p_{j 1})$ for some $u \in A_i$ and $w \in B_j$ so that $(u,w) \in E$ (note that some such edge must exist or else $\pi_{(i,j)}$ is empty).  This obviously creates a canonical path that spans $(s^p_{i k_A}, t^p_{j k_B})$ (since $E_M \subseteq H_1$) while costing at most $2|H|$, and thus $H'$ is a valid $k$-spanner of size at most $6 |H|$ as claimed.
\end{proof}

We now can prove one direction of the reduction:

\begin{lemma} \label{lem:reduction1}
Given a $k$-spanner $H$ for $G'$, we can construct in polynomial time a REP-cover for $G$ of size at most $6|H| / x$
\end{lemma}
\begin{proof}
We first apply Lemma~\ref{lem:proper} to get a proper $k$-spanner $H'$ of size at most $6 |H|$.  Now for every $p \in [x]$ and $i \in A$ let $S^p_i = \{u \in A_i : (s^p_{i1}, u) \in E(H')\}$, and similarly for $j \in B$ let $T^p_j = \{w \in B_j : (w, t^p_{j1}) \in E(H')\}$.  Now for each $p \in [x]$ let $U^p = (\cup_{i \in A} S^p_i) \cup (\cup_{j \in B} T^p_j)$.  We claim that for every $p \in [x]$, the set $U^p$ is a valid REP-cover.  This is enough to prove the lemma, since clearly the smallest $U^p$ has size at most $6|H| / x$ by averaging (each vertex in each $U^p$ can be charged to the edge in $E_{sA}$ or $E_{tB}$ that caused it to be in $U^p$).

Consider a superedge $(i,j) \in \tilde E$.  Since $H'$ is a proper $k$-spanner, it contains a canonical path from $s^p_{i k_A}$ to $t^p_{j k_B}$.  By the definition of canonical path, this implies that it includes the edges $(s^p_{i1}, u), (u, w)$, and $(w, t^p_{j1})$ for some $u \in A_i$ and $w \in B_j$.  Thus $u \in S^p_i$ and $w \in T^p_j$ with $(u,w) \in E$.  Since this is true for every superedge, it implies that $U^p$ is a valid REP-cover.
\end{proof}

We now want to prove the other direction, that the existence of a small REP-cover for $G$ implies a small $k$-spanner for $G'$.

\begin{lemma} \label{lem:reduction2}
Given a REP-cover $C$ for $G$, we can construct in polynomial time a $k$-spanner $H$ of $G'$ with at most $(k+1)x|C|$ edges.
\end{lemma}
\begin{proof}
The edge set of our $k$-spanner $H$ is
\begin{equation*}
\{(s^p_{i1}, u) : i \in A, u \in A_i \cap C, p \in[x]\} \cup \{(t^p_{j1}, w) : j \in B, w \in B_j \cap C, p \in [x]\} \cup E \cup E_M \cup \hat E.
\end{equation*}

To see that $H$ is a $k$-spanner, first note that every edge in $E \cup E_M \cup \hat E$ is spanned by itself.  It is easy to see that any edge in $E_{sA}$ or $E_{tB}$ is spanned by $\hat E$ (in fact, $3$-spanned).  And since $C$ is a valid REP-cover, for any edge in $E_{\tilde G}$ there is a canonical path included in $H$.

Now we need to bound the size of $H$.  Clearly it is at most
\begin{align*}
|H| &= x|C| + |E| + |E_M| + |\hat E|\\
&\leq x|C| + n^2 + x|A|(k_A-1) + x|B|(k_B-1) + (n + x \tilde n)\\
& \leq x|C| + x \tilde n + x(\tilde n / 2)(k_A-1) + x (\tilde n / 2)(k_B-1) + n + x \tilde n \\
&\leq x|C| + x|C| + x(\tilde n / 2)(k-3) + x \tilde n + x \tilde n \\
& \leq 4 x |C| + x |C| \frac{k-3}{2} \\
& \leq (k+1)x|C|,
\end{align*}
where we used the fact that $|C| \geq \tilde n$ and that by definition $x = n^2 / \tilde n$
\end{proof}

We can now prove the main theorem of the paper, that it is hard to approximate the basic $k$-spanner problem.

\begin{theorem} \label{thm:spanner_hardness}
Assuming $NP \not\subseteq BPTIME(2^{polylog(n)})$, for any constant $\epsilon > 0$ and $3 \leq k \leq \log^{1-2\epsilon} n$  there is no polynomial time approximation algorithm for the basic $k$-spanner problem with ratio less than $2^{(\log^{1-\epsilon} n) / k}$.
\end{theorem}
\begin{proof}
Suppose that we have an $\alpha$-approximation algorithm for the basic $k$-spanner problem.  Given an instance $G$ of Min-Rep with supergirth larger than $k+1$, we reduce it to basic $k$-spanner on $G'$ as described above.  If the smallest REP-cover has size $\tilde n$ (i.e.\ we can assign a single label to every vertex of the supergraph and get a valid proof), then by Lemma~\ref{lem:reduction2} there is a $k$-spanner of $G'$ with at most $(k+1)x \tilde n$ edges.  On the other hand, if the smallest REP-cover has size at least $2^{(\log^{1-\epsilon} n) / (k+1)} \cdot \tilde n$ then by Lemma~\ref{lem:reduction1} the smallest $k$-spanner of $G'$ must have size at least $2^{(\log^{1-\epsilon} n) / (k+1)} \cdot \tilde n x / 6$.  By Theorem~\ref{maintheorem} we cannot distinguish between these two cases of Min-Rep, and thus $\alpha \geq (2^{(\log^{1-\epsilon} n) / (k+1)} \tilde n x / 6) / ((k+1)x \tilde n) = (2^{(\log^{1-\epsilon} n) / (k+1)}) / (6(k+1))$.

However, the $n$ used in the above expression is the size of the Min-Rep instance $G$, not the size of the spanner instance $G'$.  Let $n' = |V'|$ be the size of the $k$-spanner instance, and note that $n' \leq n + \tilde n k x = n + n^2 k \leq 2k n^2$.  So $n \geq \sqrt{n'/(2k)}$, and thus we have hardness
\begin{equation*}
\frac{2^{(\log^{1-\epsilon} n) / (k+1)}}{6(k+1)} \geq \frac{2^{(\log^{1-\epsilon} (n'/(2k))) / (2(k+1))}}{6(k+1)}
\end{equation*}

By using an appropriately smaller $\epsilon$, and switching notation to let $n$ represent the size of the $k$-spanner instance, this gives hardness of $2^{(\log^{1-\epsilon} n) / k}$ as claimed (assuming that $k \leq \log^{1-2\epsilon} n$).
\end{proof}

\section{Conclusion}

Motivated by proving hardness for the basic $k$-spanner problem (one of the only spanner problems for which strong hardness was not already known), we gave a proof that Label Cover and Min-Rep are hard to approximate even when restricted to instances with large supergirth.  This result has been claimed before~\cite{ELD}, but their proof was fundamentally flawed by their attempt to increase the girth \emph{before} using parallel repetition.   Our new proof is based on a technique (subsampling edges of $2$-query PCPs) that allows us to sparsify the PCP obtained by parallel repetition enough to destroy most small cycles without significantly losing in the soundness of the PCP (and thus the provable hardness).  This gives a proof that the basic $k$-spanner problem, which is perhaps the simplest of spanner problems, has superpolylogarithmic hardness.

\paragraph{Min-Rep with large girth and large gap?}
An important but perhaps difficult question is if Min-Rep is still hard to approximate on instances with large girth (even if the supergirth is $4$). A solution to this question
would be useful in lower bounding problems such as Multicommodity Buy-at-Bulk,
Multicommodity Cost-Distance (see \cite{STOL}), and other network design problems, and would also lead to simplifications of already known lower bounds.

It may not be possible to obtain hardness for the problem of Min-Rep with large girth,
and in fact it may be that this problem has a good approximation.  For example, one natural way to try to prove this would be to find instances where the girth is at least the supergirth, and then try to repeat our reduction.  A natural case where the girth is at least the supergirth is if the graph induced by $A_i \cup B_j$ is a matching for each superedge $(i,j)$.  The Label Cover version (i.e.~maximization) of this is the Unique Games problem, originally defined by Khot \cite{khot}.  It is easy to see that in this case the girth is at least the supergirth, but unfortunately while Khot conjectured the problem to be hard to approximate (the famous \emph{Unique Games Conjecture}), the conjectured hardness is still quite small and known upper bounds preclude the kind of strong hardness that we would like to prove.  But are there other cases where the girth of the Min-Rep graph is large that are more difficult to approximate?



\bibliographystyle
{abbrv}
{\bibliography{z}
\end{document}